\documentclass[11pt,letterpaper]{article}

\usepackage[margin=1in]{geometry} 
\usepackage{amsmath,amsthm,amssymb}
\usepackage{thm-restate}
\usepackage[utf8]{inputenc}
\usepackage{color} 
\usepackage[numbers]{natbib}
\usepackage{graphicx}
\usepackage{wrapfig} 
\usepackage{float} 
\usepackage[inline]{enumitem} 
\usepackage[bookmarks=false]{hyperref}
\usepackage{cleveref}

\usepackage{algpseudocode,algorithm,algorithmicx}

\newcommand*{\rev}[1]{{#1}}  

\newcommand*{\whp}{w.h.p.\ }
\newcommand*\FnName[1]{\texttt{#1}}
\newcommand*{\tO}{\widetilde{O}}
\newcommand*{\tOmega}{\widetilde{\Omega}}
\newcommand*{\tTheta}{\widetilde{\Theta}}

\newcommand{\mc}{\mathcal}
\newcommand{\eps}{\varepsilon}
\newcommand{\poly}{\mathrm{poly}}

\newtheorem{lemma}{Lemma}
\newtheorem*{lemma*}{Lemma}
\newtheorem*{theorem*}{Theorem}
\newtheorem{corollary}{Corollary}
\newtheorem*{corollary*}{Corollary}
\newtheorem{definition}{Definition}
\newtheorem{theorem}{Theorem}
\newtheorem{goal}{Goal}

\newcommand{\FullOrShort}{full}
\ifthenelse{\equal{\FullOrShort}{full}}{
  \newcommand{\fullOnly}[1]{#1}
  \newcommand{\shortOnly}[1]{}
}{

  \newcommand{\fullOnly}[1]{}
  \newcommand{\shortOnly}[1]{#1}
}

\begin{document}
        \title{Low-Congestion Shortcuts without Embedding\footnotemark[1]}
        \author{Bernhard Haeupler\footnotemark[2], Taisuke Izumi\footnotemark[3], Goran Zuzic\footnotemark[2]} 
        \date{}
        \maketitle
        
        \renewcommand{\thefootnote}{\fnsymbol{footnote}}

        \footnotetext[1]{This is a post-peer-review, pre-copyedit version of an article published in Distributed Computing (DIST). The final authenticated version is available online at: http://dx.doi.org/10.1007/s00446-020-00383-2. This work was supported in part by KAKENHI No. 15H00852 and 16H02878 as well as NSF grants CCF-1527110, CCF-1618280, CCF-1814603, CCF-1910588, NSF CAREER award CCF-1750808, a Sloan Research Fellowship and the 2018 DFINITY Scholarship.}

        \footnotetext[2]{Carnegie Mellon University, Pittsburgh PA, USA. E-mail: \{haeupler,gzuzic\}@cs.cmu.edu.}

        \footnotetext[3]{Nagoya Institute of Technology, Gokiso-cho, Showa-ku, Nagoya, Aichi, Japan. E-mail: t-izumi@nitech.ac.jp.}

        \renewcommand{\thefootnote}{\arabic{footnote}}
        \setcounter{footnote}{0}

\begin{abstract}
Distributed optimization algorithms are frequently faced with solving sub-problems on disjoint connected parts of a network. Unfortunately, the diameter of these parts can be significantly larger than the diameter of the underlying network, leading to slow running times. This phenomenon can be seen as the broad underlying reason for the pervasive $\tilde{\Omega}(\sqrt{n} + D)$ lower bounds that apply to most optimization problems in the CONGEST model. On the positive side, [Ghaffari and Hauepler; SODA'16] introduced low-congestion shortcuts as an elegant solution to circumvent this problem in certain topologies of interest. Particularly, they showed that there exist good shortcuts for any planar network and more generally any bounded genus network. This directly leads to fast $O(D \log^{O(1)} n)$ distributed algorithms for MST and Min-Cut approximation, given that one can efficiently construct these shortcuts in a distributed manner.

Unfortunately, the shortcut construction of [Ghaffari and Hauepler; SODA'16] relies heavily on having access to a genus embedding of the network. Computing such an embedding distributedly, however, is a hard problem---even for planar networks. No distributed embedding algorithm for bounded genus graphs is in sight.

In this work, we side-step this problem by defining \rev{tree-restricted shortcuts:} a more structured and restricted form of shortcuts. We give a novel construction algorithm that efficiently finds such shortcuts that are, up to a logarithmic factor, as good as the best \rev{restricted} shortcuts that exist for a given network. This new construction algorithm directly leads to an $O(D \log^{O(1)} n)$-round algorithm for solving optimization problems like MST for any topology (i.e., underlying graph) for which good \rev{restricted} shortcuts exist---without the need to compute any embedding. \rev{This greatly simplifies the existing planar algorithms and includes the first efficient algorithm for bounded genus graphs.}
\end{abstract}

\newpage

\section{Introduction}

\subsection{Background and motivation}

Consider the problem of finding the Minimum Spanning Tree (MST) on a distributed network with $n$ independent processing nodes.
The network is abstracted as a graph $G = (V, E_G)$ with $n$ nodes and diameter $D$.
The nodes communicate by synchronously passing $O(\log n)$-bit messages to each of its direct neighbors.
The goal is to design algorithms (protocols) that minimize the number of synchronous message passing rounds before the nodes collaboratively solve the optimization problem.

The message-passing setting we just described is a model called CONGEST~\cite{Peleg:2000}.
The MST problem can be solved in such a setting using $O(\sqrt{n}\log^* n + D)$ rounds of communication~\cite{Kutten-Peleg}\footnote{\rev{The algorithm can be easily modified to run in $O(\sqrt{n}\log^* n + D)$ rounds of communication by growing components to size $\sqrt{n / \log^* n}$ in the first phase of the algorithm.}}.
Moreover, and perhaps more surprisingly, this bound was shown to be the best possible (up to polylogarithmic factors).
Specifically, there are graphs in which one cannot do any better than $\tOmega(\sqrt{n} + D)$~\cite{Peleg-Rubinovich-1999, elkin2006unconditional, DasSarma-11}\footnote{Throughout this paper, $\tO(\cdot)$, $\tTheta(\cdot)$, and $\tOmega(\cdot)$ hide polylogarithmic factors in $n$, the number of nodes in the network.}.
While clearly no algorithm can solve any global network optimization problem faster than $\Omega(D)$, the $\tOmega(\sqrt{n})$ factor is harder to discern.
To make matters worse, the $\tOmega(\sqrt{n} + D)$ lower bound was shown to be far-reaching.
It applies to a multitude of important network optimization problems including MST, minimum-cut, weighted shortest-path, connectivity verification and so on~\cite{DasSarma-11}

While this bound precludes the existence of more efficient algorithms in the general case, it was not clear whether it holds for special families of graphs.
This question is especially important because any real-world application on huge networks should exploit the special structure that the network provides.
The mere existence of ``hard'' networks for which one cannot design any fast algorithm might not be a limiting factor.

In the first result that utilizes network topology (i.e., the structure of the communication graph) to circumvent the lower bound, Haeupler and Ghaffari designed an $\tO(D)$-round distributed MST algorithm for planar graphs~\cite{gh2016lowcongestion}. Note that this algorithm offers a huge advantage over older results for planar graphs with small diameters.

They achieve this by introducing an elegant abstraction for designing distributed algorithms named \textbf{low-congestion shortcuts}.
Their methods could in principle be used to achieve a similar result for genus-bounded graphs, but their presented algorithms have a major technical obstacle: they require a surface embedding of the planar/genus bounded graph to construct the low-congestion shortcuts. While computing a distributed embedding for planar graphs has a complex $\tO(D)$-round solution~\cite{GH15-Embedding}, this remains an open problem for genus-bounded graphs~\cite{gh2016lowcongestion}.

This paper side-steps the issue by vastly simplifying the construction of low-congestion shortcuts.
We define a more structured version of low-congestion shortcuts called \textbf{tree-restricted shortcuts} and propose a simple and general distributed algorithm for finding them.
\rev{On many graphs of interest these shortcuts are as powerful as the general ones (see the discussion in \Cref{sec:subsequent-work} for a short comparison).}
Moreover, the algorithm is completely oblivious to any intricacies of the underlying topology and finds universally near-optimal tree-restricted shortcuts.
As a simple consequence of our construction technique, we get an $\tO(gD)$-round algorithm for genus $g$ graphs, \rev{a result that was not known before the conference version of this paper was published}.
We believe that this simplicity makes the algorithm usable even in practice.

\subsection{A brief overview of low-congestion shortcuts}
We now give a short introduction to the general low-congestion shortcut framework, as defined in \cite{gh2016lowcongestion}.
Consider the following recurring scenario throughout many distributed optimization problems:
\begin{definition}[Part-wise aggregation]
  \textit{Let $G = (V, E_G)$ be a graph. Given disjoint and internally-connected \textbf{parts} $P_1, P_2, \ldots, P_N \subseteq V$, we want to distributedly compute some simple part-wise aggregate (e.g., sum or max) of nodes' private values. Specifically, each node is initially given its part ID (or $\bot$ if none) and a private value $x_v$; at the end of the computation each node $v$ belonging to some part $P_i$ should know the aggregate value of $\{ x_v \mid v \in P_i \}$.}
\end{definition}

A classical example for such a scenario is the 1926 algorithm of Boruvka~\cite{Boruvka26} for computing the MST: We start with a trivial partition of singleton parts for each node. For $O(\log n)$ iterations each part computes the minimum-weighted outgoing edge, adds it to the MST, and merges with the other part incident to this edge.

A key concern in designing a distributed version of Boruvka's algorithm is finding good communication schemes that allow the nodes of some part to collaborate without interfering with other parts. While a natural solution would be to allow communication only inside the same part (which is feasible since the parts are internally connected), this could take a long time. The problem appears when the diameter of a part in isolation is much larger than the diameter $D$ of the original graph $G$.

\textbf{Low-congestion shortcuts}~\cite{gh2016lowcongestion} were introduced to overcome this issue: each part $P_i$ is allowed to use a set of extra edges $H_i \subseteq E_G$ to more efficiently communicate with other nodes in the same part. More precisely, part $P_i$ is permitted to use the edges $E_G[P_i] \cup H_i$ for communication, where $E_G[P_i]$ are edges with both endpoints in $P_i$.

We say that a shortcut has \textbf{dilation} $d$ if the diameter of $E_G[P_i] \cup H_i$ is at most $d$ for all parts. Similarly, it has \textbf{congestion} $c$ when each edge is assigned to at most $c$ different parts. We give the formal definitions below.
\begin{definition}
  Let $G = (V, E_G)$ be an undirected graph with vertices subdivided into \textbf{disjoint and connected} subsets $\mathcal{P} = (P_1, P_2, ..., P_N), P_i \subseteq V$. In other words, $E_G[P_i]$ is connected and $P_i \cap P_j = \emptyset$ for $i \neq j$.
  The subsets $P_i$ are called \textbf{parts}.
  We define a \textbf{shortcut} $\mathcal{H}$ as $(H_1, H_2, ..., H_N)$, $H_i \subseteq E_G$. A shortcut is characterized by the following parameters:
  \begin{enumerate}
  \item $\mathcal{H}$ has congestion $c$ if each edge $e \in E_G$ is used in at most $c$ different sets $E_G[P_i] \cup H_i$, i.e.,  $\forall e \in E_G: \ |\{i : e \in E_G[P_i] \cup H_i \}| \le c$. Note that the sets $\{ E_G[P_i] \}_{i=1}^N$ are disjoint.
  \item $\mathcal{H}$ has dilation $d$ if for each $i \in [N]$ the diameter of $E_G[P_i] \cup H_i$ is at most $d$.
  \end{enumerate}
\end{definition}

Finally, we define the \textbf{quality} $q$ of a shortcut as $\mathrm{congestion} + \mathrm{dilation}$, a classic parameter extensively used in routing~\cite{LMR94-routing}.

If we can efficiently construct shortcuts with quality $q$, we can solve problems such as MST and approximate Min-Cut in $\tO(q)$ rounds~\cite{gh2016lowcongestion}. One would ideally want $\tO(D)$-quality shortcuts since going below the diameter is clearly impossible for global problems such as the MST or Min-Cut, since otherwise two nodes at distance $D$ apart would not be able to exchange any information about themselves. However, the pervasive $\tOmega(\sqrt{n} + D)$ lower bound implies we cannot find shortcuts with $q = \mathrm{congestion} + \mathrm{dilation} = \tO(D)$ on general graphs, for many graph families shortcuts of quality $\tO(D)$ exist. For example, planar graphs always offer (optimal) $\tO(D)$-quality shortcuts, and such shortcuts can be found in $\tO(D)$ rounds, thus bypassing the $\tOmega(\sqrt{n} + D)$ lower bound.


\subsection{Our contribution}\label{sec:contrib}

Roughly speaking, there are two challenges in the design of shortcut-based algorithms. Let $\mathcal{G}$ be the target class of graphs we want to design distributed algorithms. The first challenge is to identify the optimal (smallest) value $q$ such that $\mathcal{G}$ has shortcuts of quality $q$. This is purely a graph-theoretic problem. The second challenge is to convert the existential result proved by the first challenge to the constructive result, i.e., we must design a distributed algorithm constructing efficient shortcuts for that class. This is a distributed computing problem that might be distinctively harder than the former one. Indeed, while one can prove that bounded genus graphs have good-quality shortcuts, the proof is not constructive because it requires access to an embedding~\cite{gh2016lowcongestion}; this is the primary reason why fast algorithms for bounded genus graphs were not known. Even in the planar case, distributedly constructing such an embedding is known, but complicated.

A natural idea to simplify algorithm design would be to come up with a generic procedure that finds a shortcut of quality $q$ for the best (or approximately best) $q$. Such a result would automatically lift a purely existential result to a constructive one. However, such a result is currently unknown and is the central (open) problem in the area of low-congestion shortcuts.
\begin{goal}\label{maingoal}
  Let $\mathcal{P} = (P_1, \ldots, P_N)$ be a set of parts in a graph $G$. Distributedly construct shortcuts of quality $\tO(q)$ in $\tO(q)$ rounds, where $q$ is the optimal shortcut quality (with respect to $\mathcal{P}$).
\end{goal}

We resolve the above question for some important classes of graphs. We introduce a more structured definition of shortcuts called \textbf{tree-restricted shortcuts} and give a constructive algorithm that finds the nearly optimal tree-restricted shortcuts in any graph that contains them. While the new shortcut definition is a strict subset of the old definition, we leverage them to design optimal $\tO(D)$ round distributed algorithms for many graphs of interest (e.g., all planar graphs and all bounded genus graphs).


\smallskip

The details of our contribution are summarized as follows:
\begin{itemize}
\item In \Cref{sec:treerestricted}, we introduce tree-restricted shortcuts, which can only use edges of some fixed spanning tree $T \subseteq G$. Such shortcuts are characterized by congestion $c$ and \textbf{block parameter} $b$ (which substitutes the classic dilation parameter). The block parameter is  more appropriate for tree-restricted shortcuts due to their highly-structured nature: in particular, the new parameter is stronger in the sense that it implies an upper bound of $O(bD)$ on the dilation. The block parameter (upper-)bounds the number of components of $P_i$, where two nodes are in different components if they cannot reach each other via $H_i$. In \Cref{subsec:routing-on-tree-restricted-shortcuts} we propose deterministic algorithms for broadcast, convergecast, and leader election (for all parts in parallel) utilizing tree-restricted shortcuts. These yield a $O(b(D+c))$ round solution to the part-wise aggregation problem (assuming constructed tree-restricted shortcuts), a solution simpler and often faster as compared to the general-case randomized algorithms from~\cite{gh2016lowcongestion}.
\item In \Cref{sec:algorithm}, we present a generic algorithm for constructing tree-restricted shortcuts. Given a spanning tree $T$, we can find near-optimal $T$-restricted shortcuts, as formalized in the following statement.
  \begin{theorem}
    \label{theorem:shortcuts-given-guarantee}
    Let $\mathcal{P} = (P_1, \ldots, P_N)$ be parts in the graph $G$ with a spanning tree $T \subseteq G$ such that there exists a $T$-restricted shortcut with congestion $c$ and block parameter $b$. There exists a randomized distributed CONGEST algorithm that finds a $T$-restricted shortcut with congestion $O(c \log N)$ and block parameter $3b$. The shortcut can be found in $\tO(b(D + c))$ rounds.
  \end{theorem}
  Notably, when a tree-restricted shortcut with parameters $b = \tO(1)$ and $c = \tO(D)$ exists, our construction yields $\tO(D)$-quality shortcuts (since dilation is at most $O(bD)$) and, by extension, (optimal) $\tO(D)$-round algorithms for MST and approximate Min-Cut.

  \textbf{Note: } The algorithm does not know the values of $b$ and $c$ upfront if one is willing to suffer a $\tilde{O}(1)$ performance hit. In particular, it is possible to find a feasible pair $(b, c)$ that yields a near-optimal value of $b(D+c)$. Given an arbitrary $Q > 0$, one can check if there exists a valid pair of parameters $(b, c)$ that yield a running time of at most $\tilde{O}(b(D+c)) \le Q$. This is done by trying all $O(\log n)$ possible powers-of-two $b$ that guarantee $\tilde{O}(bD) \le Q$ and $\tilde{O}(bc) \le Q$ and truncating the execution after $Q$ rounds. Given this procedure, one can search for the smallest power-of-two $Q$ for which the above procedure succeeds (by checking all $O(\log n)$ possibilities).

  %
%
%
%
\item The final question we tackle is what graph families admit good-quality tree-restricted shortcuts.
  Fortunately, one can reinterpret prior work in the novel terminology of tree-restricted shortcuts to conclude that (any $O(D)$-depth spanning tree of) genus-$g$ graphs contain tree-restricted shortcuts with congestion $O(gD\log D)$ and block parameter $O(\log D)$. In \Cref{sec:main-results-and-applications}, we can obtain a distributed algorithm that constructs a tree-restricted shortcut with congestion $O(gD\log D\log N)$ and block parameter $O(\log D)$ for graphs with genus at most $g$. For bounded genus graphs (i.e. $g = O(1)$), the algorithms based on our shortcut construction achieves near-optimal time complexity (up to a polylogarithmic factor).
\end{itemize}

\subsection{Subsequent work: a short survey}\label{sec:subsequent-work}

Significant progress has been made since the initial conference version of this paper was published~\cite{haeupler2016low}. Subsequent work has expanded on the utility of the framework by extending it to new graph classes, new problems, and provided better construction guarantees. We intend this section to serve as a short and convenient survey of the tree-restricted shortcut framework.

\paragraph{\rev{Tree-restricted} shortcut quality and construction.} For a spanning tree of depth $O(D)$, we define the \textbf{T-quality} (denoted $q_{T}$) of a $T$-restricted shortcut as $q_{T} := b D + c$ (where $b$ is the block parameter and $c$ is the congestion). This definition is simply the $\mathrm{congestion} + \mathrm{dilation}$, i.e. quality, when one upper-bounds the dilation as $O(bD)$ (see \Cref{sec:treerestricted} for a proof of this fact).

T-quality combines the congestion and the block parameter into a single value that sufficiently describes the shortcut construction and routing performance without the need to keep track of multiple parameters.
\begin{definition}
  A graph $G = (V, E_G)$ of diameter $D$ \textbf{admits} tree-restricted shortcuts of T-quality $q_T$ if for each spanning tree $T$ of depth $O(D)$ and each set of disjoint and connected parts $(P_i \subseteq V)_{i=1}^N$ there exists a $T$-restricted shortcut of congestion $c$ and block parameter $b$ satisfying $b\cdot D + c \le q_T$.
\end{definition}
It is not hard to see that if one can efficiently construct shortcuts of T-quality $q_T$, then a randomized algorithm can solve the part-wise aggregation problem in $\tilde{O}(q_T)$ rounds using standard random delay ideas~\cite{gh2016lowcongestion}. However, the key benefit of using the tree-restricted shortcut framework (as opposed to the general shortcut framework) is that near-optimal tree-restricted shortcuts can be efficiently and distributedly constructed.
\begin{theorem*}[Theorem 1.2 of \cite{haeupler2018round}]
  Suppose that a graph $G = (V, E_G)$ admits tree-restricted shortcuts of T-quality $q_T$. There exists a distributed CONGEST algorithm that finds a $T$-restricted shortcut with T-quality $\tilde{O}(q_T)$ in $\tilde{O}(q_T)$ rounds and sends at most $\tilde{O}(|E_G|)$ messages during its execution with high probability (with probability at least $1 - n^{-O(1)}$, where any constant can be chosen in the exponent). Moreover, the algorithm does not need to know the value of $q_T$ upfront.
\end{theorem*}
Note: We slightly reworded the main Theorem of \cite{haeupler2018round}. An appealing property of the tree-restricted shortcut framework (shared between this and subsequent work) is that one does not need to know the optimal tree-restricted shortcut T-quality $q_T^*$ upfront. This can often yield much better shortcuts than guaranteed by the theoretical bound, a property often desired in practical applications. While the paper typically assumes the algorithm knows the congestion $c$ and block parameter $b$, one can circumvent this issue with a simple exponential parameter search like the one described in \Cref{sec:contrib}.


\paragraph{Comparing \Cref{theorem:shortcuts-given-guarantee} and \cite{haeupler2018round}. } Notably, the construction of \cite{haeupler2018round} (unlike \Cref{theorem:shortcuts-given-guarantee}) controls the number of messages throughout the algorithm. Furthermore, it completes in $\tilde{O}(q_T) = \tilde{O}(bD + c)$ rounds, while the construction of \Cref{theorem:shortcuts-given-guarantee} takes $\tilde{O}(b(D + c))$ rounds. The latter result is significantly slower when $b = \log^{\omega(1)} n$, in e.g., genus- or treewidth-bounded graphs with super-polylogarithmic genus or treewidth (see \Cref{table:upper-lower-bounds-for-shortcuts} below). Furthermore, the results of \cite{haeupler2018round} can be made deterministic (with slightly worse guarantees, see below).

\paragraph{Deterministic construction.} Many of the aforementioned randomized results can be recovered in the deterministic setting while suffering only a small performance penalty. Notably, one can still construct near-optimal tree-restricted shortcuts and solve the part-wise aggregation problem in $\tilde{O}(b(D+c))$ rounds instead of $\tilde{O}(q_T) = \tilde{O}(bD + c)$ rounds (as guaranteed by the randomized procedure), even while controlling the message complexity.
\begin{theorem*}[Deterministic construction of \cite{haeupler2018round}]
  Suppose that a spanning tree $T$ of a graph $G = (V, E_G)$ admits tree-restricted shortcuts of congestion $c$ and block parameter $b$. There exists a deterministic distributed CONGEST algorithm that finds a $T$-restricted shortcut of congestion $\tilde{O}(c)$ and block parameter $\tilde{O}(b)$ in $\tilde{O}(b(D+c))$ rounds and $\tilde{O}(|E_G|)$ messages. Furthermore, one can solve the part-wise aggregation problem with the same guarantees.
\end{theorem*}

\paragraph{Graph families.} Various graph families admit good-quality tree-restricted shortcuts. \Cref{table:upper-lower-bounds-for-shortcuts} lists the known results. The last row of the table references graphs that exclude $\delta$-dense minors, meaning that all minors of $G$ have density (i.e., the ratio between the number of edges and vertices) at most $\delta$. We note that the result of \cite{ghaffari2020excluding} implies all other known upper bounds in the table (up to logarithmic factors): for instance, minor-excluded families have $\delta = O(1)$.\footnote{The excluded-dense-minor result of \cite{ghaffari2020excluding} improves the best known quality of tree-restricted shortcuts in minor-excluded graph families from $\tO(D^2)$ (proved in \cite{haeupler2018minor}) to $\tO(D)$.}

\begin{table}
  \centering
  \begin{tabular}{l l l l l l}
    \textbf{Graph Family} \
    & \multicolumn{4}{c}{\textbf{Tree-Restricted Shortcut Parameters }}
    & \textbf{Lower Bound} \\
    & & Block & Congestion & T-quality & $\Omega(d + c)$ \\
    \hline
    \rev{General} & \cite{gh2016lowcongestion} & 1~\footnotemark & $O(\sqrt{n})$ & $O(D + \sqrt{n})$ & $\tOmega(D + \sqrt{n})$ \\
    Pathwidth $k$ & \cite{haeupler2016near} & $O(k)$ & $O(k)$ & $O(kD)$ & $\Omega(kD)$ \\    
    Treewidth $k$ & \cite{haeupler2016near} & $O(k)$ & $O(k \log n)$ & $O(kD + k\log n)$ & $\Omega(kD)$ \\
    Genus $g$ & \cite{haeupler2016near} & $O(\sqrt{g})$ & $O(\sqrt{g}D\log D)$ & $O(\sqrt{g}D\log D)$ & $\Omega(\frac{\sqrt{g}D}{\log g})$ \\
    Planar & \cite{gh2016lowcongestion} & $O(\log D)$ & $O(D \log D)$ & $O(D\log D)$ & $\Omega(D \frac{\log D}{\log \log D})$ \\
    Minor-excluded & \cite{ghaffari2020excluding} & $O(1)$ & $O(D \log n)$ & $O(D \log n)$ & trivial $\Omega(D)$ \\
    No $\delta$-dense minors & \cite{ghaffari2020excluding} & $O(\delta)$ & $O(\delta D \log n)$ & $O(\delta D \log n)$ & $\Omega(\delta D)$ \\
  \end{tabular}
  \caption{Upper and lower bounds for \rev{tree-restricted} shortcuts.}
  \label{table:upper-lower-bounds-for-shortcuts}
\end{table}

\footnotetext{\rev{For general graphs, each part of size $|P_i| \ge \sqrt{n}$ is assigned the entire tree; giving them a block param. of $1$ and congestion of at most $\sqrt{n}$. Smaller parts can be handled separately in $\tilde{O}(\sqrt{n})$ rounds by using intra-part edges.}}

\paragraph{Applications.} Numerous distributed optimization tasks can be simplified and optimized by utilizing the part-wise aggregation primitive as a black-box subroutine. Applications include the MST, approximate Min-Cut, and approximate single-source shortest path (SSSP)~\cite{gh2016lowcongestion,haeupler2018faster,haeupler2018round}.

\begin{corollary}
  Suppose that a graph $G$ admits tree-restricted shortcuts of T-quality $q_T$. One can compute an (exact) MST in $\tilde{O}(q_T)$ rounds and $\tilde{O}(m)$ messages with high probability.
\end{corollary}

\rev{As a reminder, in the Min-Cut problem, one is given a graph $G = (V, E_G)$ with integer weights $w : E_G \to [1, \poly(n)]$ and needs to compute a set of edges $F \subseteq E_G$ that disconnect $G$ into at least $2$ components while minimizing the sum $\sum_{e \in F} w_e$. An $\alpha$-approximation to Min-Cut finds a set of edges that disconnects the graph whose aggregate weight is at most a multiplicative $\alpha$ factor larger than the optimal value.}

\begin{corollary}
  Suppose that a graph $G$ admits tree-restricted shortcuts of T-quality $q_T$. One can compute an $(1+\varepsilon)$-approximate (weighted) Min-Cut in $\tilde{O}(q_T) \cdot \poly(1/\eps)$ rounds and $\tilde{O}(m) \cdot \poly(1/\eps)$ messages with high probability.
\end{corollary}

\rev{In the Single-Source Shortest Path (SSSP), one is given a graph $G = (V, E_G)$ with integer weights $w : E_G \to [1, \poly(n)]$, a source $s \in V$, and needs to compute a spanning tree $T \subseteq E_G$ such that for each node $u$ we have that $d_T(s, u) = d(s, t)$ where $d(u, v)$ is the distance between $u, v \in V$ in $G$ with respect to the weight $w$, and $d_T(u, v)$ is their distance in the tree (with respect to $w$). An $\alpha$-approximation to SSSP requires the tree to satisfy $d_T(u, v) \le \alpha \cdot d(u, v)$ (note that the inequality $d_T(u, v) \ge d(u, v)$ is always satisfied).}

\begin{corollary}
  Suppose that a graph $G = (V, E_G)$ admits tree-restricted shortcuts of T-quality $q_T$. Each edge $e \in E_G$ has a weight $w_e$, and let $L$ be the weight-diameter of $G$. For any \rev{$\beta = (\log n)^{-\Omega(1)}$} one can compute an $L^{O(\log \log n) / \log(1 / \beta)}$-approximate SSSP in $\tilde{O}(q_T / \beta)$ rounds and $\tilde{O}(m / \beta)$ messages with high probability.
\end{corollary}

\rev{For instance, in the above corollary, setting $\beta = n^{-\eps}, \beta = 2^{-\Theta(\sqrt{n})}$, and $\beta = \log^{-\Theta(1/\eps)} n$ for a constant $\eps > 0$ one obtains a $\log^{O(1)} n$, $2^{\sqrt{\log n}}$, and $L^{\eps}$ approximations to SSSP, respectively.~\cite{haeupler2018faster}}

\paragraph{General shortcuts vs.~tree-restricted shortcuts.} One can easily construct pathological graph examples that admit good-quality \emph{general} shortcuts, but do not admit good-quality \emph{tree-restricted} shortcuts. \rev{For example, one can take the lower bound graph of \cite{DasSarma-11} which requires $\tOmega(\sqrt{n})$ rounds to solve MST and replace each edge with $\sqrt{n}$ parallel multi-edge copies. This immediately yields a $\tO(D) = \tO(1)$ MST solution via general shortcuts, whereas tree-restricted shortcuts are constrained by the original $\tOmega(\sqrt{n})$ lower bound.} Moreover, general shortcuts allow faster algorithms for several important graph families. For example, expander graphs and Erd\H{o}s-R\'{e}nyi random graphs admit general shortcuts of $\mathrm{dilation} + \mathrm{congestion} = \tO(1)$ for any set of parts; no such result is possible in the tree-restricted setting. However, it seems that the distributed construction of general shortcuts is a burdensome task even in highly structured graphs. The best-known result for shortcut construction and part-wise aggregation in expander graphs has round complexity $2^{O(\sqrt{\log n})} = n^{o(1)}$, significantly worse than the best existential result~\cite{ghaffari2018new}.

\subsection{Related work} \label{sec:relatedwork}

The complexity-theoretic issues in the design of distributed graph algorithms for
the CONGEST model have received much attention in the last decade. Researchers have studied many problems in-depth: Minimum-Spanning Tree~\cite{Garay-Kutten-Peleg,Kutten-Peleg,Peleg-Rubinovich-1999,Khan2008},
Maximum flow~\cite{Ghaffari:2015}, Minimum Cut~\cite{Ghaffari-Kuhn,nanongkai2014almost}, Shortest Paths, and Diameter~\cite{Nanongkai-paths,Frischknecht-Diameter-2012,Holzer-Paths-2012,lenzen2019distributed,Izumi2014}, and so on.
Most of those problems have $\tTheta(\sqrt{n} + D)$-round upper and lower 
bounds for some sort of approximation guarantee~\cite{DasSarma-11,lenzen2019distributed,Ghaffari-Kuhn,Elkin-2004,Peleg-Rubinovich-1999}. The guarantee of 
exact results sometimes yields a nearly-linear-time bound~\cite{Frischknecht-Diameter-2012}.
Note that almost all lower bounds above hold for graphs of small diameter (e.g., polylogarithmic in $n$). \rev{In such graphs we have that $\sqrt{n} \gg D$, making $\tO(D)$ algorithms strictly better than those requiring $\tO(D + \sqrt{n})$ rounds.}



\section{Preliminary: CONGEST Model}\label{sec:preliminaries}

We work in the classical CONGEST model~\cite{Peleg:2000}.
In this setting, a network is given as a connected undirected graph $G = (V, E_G)$ with diameter $D$.
Initially, nodes only know their immediate neighbors and they collaborate to compute some global function of the graph like the MST.
Communication occurs in synchronous rounds; during a round, each node can send $O(\log n)$ bits to each of its neighbors.
The nodes always correctly follow the protocol and never fail.
The goal is to design protocols that minimize the resource of time - the number of rounds before the nodes compute the solution.

We now precisely formalize the notion of solving a problem in this model, e.g.,  how are the input and output given.
While the formalization is specifically given for the MST, any other problem is completely analogous.
All nodes synchronously wake up in the first round and start executing some given protocol.
Every node initially only knows its immediate neighbors and the weight of each of its incident edges.
After a specific number of rounds, all nodes must simultaneously output
(i) the weight of the computed MST $\tau$
(ii) for each edge $e$ incident to it, a $0/1$ bit indicating if $e \in \tau$
.

\section{Tree-Restricted Shortcuts} \label{sec:treerestricted}

In this section we define tree-restricted shortcuts: a restricted version of low-congestion (i.e., general) shortcuts that are (i) simpler to work with, (ii) often equally powerful as the general shortcuts, (iii) offer deterministic routing schemes and, most importantly, (iv) can be efficiently constructed on any graph that contains them. Following the definitions, we rephrase the relevant prior work in our new term, showcase an efficient deterministic routing scheme, and finally state our main result and applications.

\subsection{Definition}

Tree-restricted shortcuts are low-congestion shortcuts with the additional property that $H_i$ is restricted to (the edges of) some spanning tree $T$. \rev{The running time of algorithms will depend on the depth of $T$, hence we will assume throughout the paper that $T$ is some tree of depth $O(D)$ (e.g., a BFS tree); the user of the framework is otherwise free to choose any tree $T$.}
\begin{definition}
  Let $\mathcal{H} = (H_1, H_2, ..., H_N)$ be a (general) shortcut on the graph $G = (V, E_G)$ with respect to the parts $\mathcal{P} = (P_i)_{i=1}^N$. Given a rooted spanning tree $T = (V, E_T) \subseteq G$ we say that a shortcut $\mathcal{H}$ is \emph{tree-restricted} or \textbf{$T$-restricted} if for each $i \in [N], H_i \subseteq E_T$ i.e.,  every edge of $H_i$ is a tree edge of $T$.
\end{definition}

Congestion and dilation are still well-defined for tree-restricted shortcuts. However, it is more convenient to use an alternative \textbf{block parameter}, \rev{which in turn also bounds the dilation. The block parameter (upper-)bounds the number of components of $P_i$, where two nodes $u, v \in P_i$ are in different components if they cannot reach each other via $H_i$.}

\begin{definition}
  Let $\mathcal{H} = (H_1, H_2, ..., H_N)$ be a $T$-restricted shortcut on the graph $G = (V, E_G)$ with respect to the parts $\mathcal{P} = (P_i)_{i=1}^N$. Fix a part $P_i$ and consider the connected components of the subgraph $(V, H_i)$. \rev{If a component contains at least one node of $P_i$, we call it a \textbf{block component} (e.g., an isolated $v \in P_i$ is a block component).} Furthermore, we say $\mc{H}$ has \textbf{block parameter} $b$ if the number of block components associated with each part is at most $b$.
\end{definition}

\begin{figure}[h]
  \centering
  \includegraphics[width=0.4\textwidth]{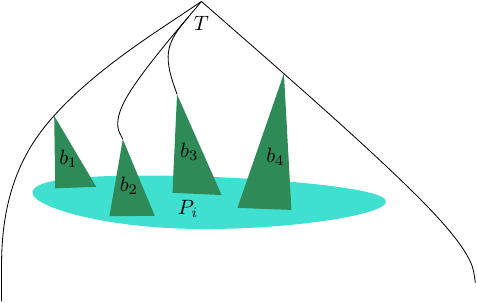}
  \caption{Illustration of a $T$-restricted shortcut subgraph for a part $P_i$, composed of block components $b_1, b_2, b_3$ and $b_4$.}
\end{figure}


\rev{Note that a connected component of $(V, H_i)$ without nodes in $P_i$ does not need to be counted; it does not need any information from the part-wise aggregation of part $i$. On the other hand, an isolated vertex $\{v\}$ where $v \in P_i$ must be counted.} \Cref{lemma:diameter-of-shortcut-with-b-blocks} argues that a block parameter of $b$ implies the dilation of $b(2\cdot\mathrm{depth}(T)+1)$. From now on, we will assume that $T$ is chosen to have depth $O(D)$, which is asymptotically minimal and achievable via a BFS tree. We note that distributedly computing a BFS tree is a classic problem with a simple $O(D)$ round CONGEST algorithm~\cite{Peleg:2000}.
\begin{lemma}
  \label{lemma:diameter-of-shortcut-with-b-blocks}
  Let $T$ be a spanning tree with depth at most $D$ and let $\mathcal{H} = (H_i : i \in [N])$ be a $T$-restricted shortcut with congestion $c$ and block parameter $b$ with respect to parts $\mathcal{P} = (P_i : i \in [N])$. Then the dilation of $\mathcal{H}$ is at most $b(2D+1)$.
\end{lemma}
\begin{proof}
  Fix $i \in [N]$. Contract every block component of $H_i$ into a supernode and remove all other nodes. This supergraph will contain $b' \le b$ supernodes and will be connected (because $E_G[P_i]$ is connected). Hence its diameter is $b'-1 \le b-1$. Every supernode corresponds to a block component of diameter $2D$, implying the diameter of $E_G[P_i] \cup H_i$ is at most $2bD + b-1 < b(2D+1)$.
\end{proof}


\subsection{Shortcuts on genus-bounded and planar graphs}
We say that a graph \textbf{admits} tree-restricted shortcuts if there always exists good tree-restricted shortcuts, even with respect to an adversarially chosen set of parts and spanning tree $T$ (of low depth).
\begin{definition}
  A graph $G = (V, E_G)$ of diameter $D$ \textbf{admits} tree-restricted shortcuts of congestion $c$ and block parameter $b$ if for each spanning tree $T$ of depth $O(D)$ and each set of disjoint and connected parts $\mathcal{P} = (P_i \subseteq V)_{i=1}^N$ there exists a $T$-restricted shortcut of congestion $c$ and block parameter $b$.
\end{definition}

Tree-restricted shortcuts are particularly useful on genus-bounded (e.g., planar) graphs. In particular, we can reinterpret the low-congestion result of Haeupler and Ghaffari~\cite{gh2016lowcongestion} using our notation.
\begin{theorem*}[Haeupler and Ghaffari~\cite{gh2016lowcongestion}]
  \label{theorem:shortcut-existence-on-bounded-genus}
  Genus-$g$ graphs admit tree-restricted shortcuts with congestion $O(gD \log D)$ and block parameter $O(\log D)$.
\end{theorem*}

\rev{We note that the paper~\cite{gh2016lowcongestion} proves the analogous claim about general shortcuts and does not explicitly talk about tree-restricted shortcuts. However, their proof implicitly argues precisely about the congestion and block parameter of tree-restricted shortcuts without explicitly referring to them. In particular, their $O(D \log D)$ dilation bound is implicitly derived by arguing about the block parameter being $O(\log D)$ and using \Cref{lemma:diameter-of-shortcut-with-b-blocks}.}
However, note that their theorem proves only the existence of such shortcuts. While the original paper does describe an algorithm that can in principle be used to compute them, it requires an embedding of $G$ on a surface of genus $g$. It is an open problem to compute such an embedding efficiently in the CONGEST model.

\subsection{Deterministic routing on tree-restricted shortcuts}\label{subsec:routing-on-tree-restricted-shortcuts}

In this section, we show how the structure of tree-restricted shortcuts can be useful in facilitating communication within parts. From a high level, the tree-like structure allows for fast, deterministic, and simultaneous broadcasting/convergecasting on block components; this can be easily extended to true part-wise aggregation. For clarity, broadcast is defined as an operation on a rooted (sub)tree that floods some value from the root down to all other nodes; convergecast is defined as an aggregation of nodes' private values starting from the leaves and towards the root (ending in the root knowing the final aggregate). \Cref{lemma:routing-on-trees} gives a way how to simultaneously perform these primitives on subtrees.
\begin{lemma}[Routing on subtrees]
  \label{lemma:routing-on-trees}
\global\def\StateLemmaRoutingOnTrees{
  Let $T$ be a rooted tree of depth $O(D)$ and let $T_1, T_2, \ldots, T_k \subseteq T$ be a family of subtrees where each edge of $T$ is contained in at most $c$ subtrees, i.e., $| \{ i \mid e \in T_i, i \in [k] \}| \le c$. There is a simple deterministic algorithm that can perform a convergecast/broadcast on all of the subtrees in $O(D + c)$ CONGEST rounds.
}\StateLemmaRoutingOnTrees
\end{lemma}
\global\def\ProofRoutingOnTrees{\begin{proof}
  We describe the convergecast algorithm. Each message sent during the algorithm will have a subtree-ID $i$ associated with it. Suppose that a node $v$ is in a subtree $T_i$ (a node can be contained in multiple subtrees). We say $(v, i)$ is active when $v$ receives a message associated with $i$ from all of its $T_i$-children (if $v$ is a leaf in $T_i$, then $(v, i)$ is immediately active). When $(v, i)$ becomes active, it will schedule an ID-$i$ message to be sent along its $T$-parent edge; note that two messages scheduled along the same edge cannot have the same ID. Each round, if multiple messages are scheduled over the same $T$-edge, the algorithm sends the message associated with the ID $i$ that minimizes $\mathrm{depth}_T(\mathrm{root}(T_i))$. Here, $\mathrm{depth}_T(v)$ is the length of the unique path between $\mathrm{root}(T_i)$ and $v$ in $T$. Ties are broken by the ID $i$ itself. The convergecast and broadcast operations are symmetric, so we will only prove the lemma for convergecasts.

  We now analyze the algorithm.
  Fix a node $v$. It is sufficient to prove that no message gets transmitted along $v$'s parent edge after $\mathrm{height}_T(v) + c = O(D + c)$ rounds where $\mathrm{height}_T(v)$ is the maximum distance between $v$ and any leaf in $T$ that is a descendant of $v$ (the unique path between the $T$-root and the leaf goes through $v$).

  Note that any message that gets transmitted along $v$'s parent edge must belong to a subtree $T_i$ that contains that edge. Let $I = (i_1, i_2, ..., i_k)$ be the IDs of subtrees that contain $v$'s parent edge, ordered by their priority (as described). In particular, we say that $T_{i_p}$ has priority $p$. The congestion condition stipulates that $k \le c$.

  We will prove by induction that for $p \in [k]$ the message associated with $i_p$ will be transmitted no later than round $\mathrm{height}_T(v) + p$. The claim clearly holds for the leaves of $T$. Note that (i) the relative priority-ordering between $I$ is unchanged with respect to any node of $T$ (other than $v$), (ii) any subtree $T_i$ that is contained in the set of descendants of $v$, but does not contain the parent edge of $v$ will have lower priority than any subtree in $I$.

  Fix $i_p$. By the induction hypothesis, messages corresponding to $\{i_1, \ldots, i_{p-1}\}$ will be sent strictly before round $\mathrm{height}_T(v) + p$. It is sufficient to argue that $v$ has received messages corresponding to $i_p$ from all of its $T_{i_p}$-children before round $\mathrm{height}_T(v) + p$. However, this can be directly argued from the induction: for any child $w \in T_{i_p}$ we have $\mathrm{height}_T(w) \le \mathrm{height}_T(v) - 1$, hence the priority of $i_p$ is at most $p$ with respect to $w$.
  Hence $v$ will send the message corresponding to $i_p$ no later than round $\mathrm{height}_T(v) + p$ and we are done.%
\end{proof}}
\fullOnly{\ProofRoutingOnTrees}
\shortOnly{\begin{proof}Deferred to \Cref{sec:routing-proofs}.\end{proof}}


Convergecast and broadcast are used to facilitate routing in tree-restricted shortcuts. We can intuitively envision the shortcut edges $H_i$ as a family of subtrees (in our notation: block components). Aggregation of values within each block component can be exactly achieved by simultaneously convergecasting and broadcasting in all block components. We extend this result to true part-wise aggregation.



\begin{theorem}[Routing on tree-restricted shortcuts]
  \label{theorem:routing-on-tree-restricted-shortcuts}
\global\def\StateTheoremRoutingOnTreeShortcuts{
  Given a $T$-restricted shortcut with congestion $c$ and block parameter $b$, there are deterministic distributed algorithms that terminate in $O(b(D + c))$ rounds for the following problems.
  \begin{enumerate}
    \item Electing a leader for each of the parts in parallel.
    \item Convergecasting $O(\log n)$-bit messages to the leader of each part in parallel.
    \item Broadcasting a $O(\log n)$-bit message from the leader of each part in parallel.
  \end{enumerate}
}\StateTheoremRoutingOnTreeShortcuts
\end{theorem}
\global\def\ProofRoutingOnTreeShortcuts{\begin{proof}
  All of these algorithms have a common flavor: for each part we perceive its shortcut edges $H_i$ as a supergraph of at most $b$ supernodes where each supernode corresponds to a block component. We proceed to describe each of the algorithms on the supergraph and implicitly assume that intra-block communication happens after each step of the algorithm.

  Communication within block components can be done in parallel using \Cref{lemma:routing-on-trees}: all the nodes of a block component convergecast the relevant information to the block-root and subsequently the block-root broadcasts the result back. 

  \textbf{Electing a leader for each part} is performed by electing a leader for each supernode (block component) and broadcasting the leader to all neighborhood supernodes for $b$ steps. Every supernode keeps the smallest leader ID ever seen as its current leader. After $b$ rounds all the supernodes have the same leader. The algorithm requires $O(b(D + c))$ rounds as each of the $b$ broadcasting steps is followed by an $O(D + c)$ intra-block communication step.

  \textbf{Broadcasting/convergecasting from/to the leader for each part} can be done by building a BFS tree from the leader-supernode. We can utilize the standard distributed BFS algorithm on the supergraph requiring $O(b)$ steps. The algorithm similarly requires $O(b(D + c))$ rounds as each of the $O(b)$ BFS steps is followed by an $O(D + c)$-round intra-block communication step.
\end{proof}}
\fullOnly{\ProofRoutingOnTreeShortcuts}
\shortOnly{\begin{proof}Deferred to \Cref{sec:routing-proofs}.\end{proof}}

We also state a simple technical lemma we use for the construction of tree-restricted shortcuts.
\begin{lemma}
  \label{lemma:det-check-if-block-at-most-b}
  \global\def\StateDetCheckBlocks{
    Given a $T$-restricted shortcut with congestion $c$, a deterministic distributed algorithm can identify all parts with at most $b'$ block components. Specifically, after the algorithm terminates each node within a part $i$ knows if $P_i$ is composed of more than $b'$ block components. The algorithm executes in $O(b'(D + c))$ rounds.
  }\StateDetCheckBlocks
\end{lemma}
\global\def\ProofDetCheckBlocks{\begin{proof}
    Similarly to the proof of \Cref{theorem:routing-on-tree-restricted-shortcuts}, for each part $P_i$ we consider the (connected) supergraph where each supernode corresponds to a block component of $H_i$. We need to find all parts whose supergraphs have at most $b'$ supernodes.

  Each supernode broadcasts its leader for exactly $b'$ rounds and every supernode keeps the minimum ID as their current leader. Subsequently, each leader $r$ (there may be multiple ones as we have not bounded the block parameter) tries to build a BFS tree comprised of all the nodes that believe $r$ is the leader. We can detect the existence of multiple leaders as in that case each BFS tree will contain two neighboring supernodes in different BFS trees and report failure. If this is not the case (all the supernodes of a part belong to the same BFS tree), we can convergecast the number of supernodes back to the root and subsequently broadcast their count back.
\end{proof}}
\fullOnly{\ProofDetCheckBlocks}
\shortOnly{\begin{proof}Deferred to \Cref{sec:routing-proofs}.\end{proof}}

\textbf{Comparison with routing on general shortcuts:} Ghaffari and Haeupler~\cite{gh2016lowcongestion} give a method for routing on general shortcuts in $O(\mathrm{dilation} \cdot \log n + \mathrm{congestion})$ rounds that is randomized and assumes a leader is already elected for each part. They describe a process of leader election via a complicated randomized bootstrapping process that takes $O(\mathrm{dilation}\cdot\log^2 n + \mathrm{congestion} \cdot \log n)$ rounds. We contrast those results with our current tree-restricted shortcut routing where leader election is simple, deterministic, and essentially no more difficult than a single convergecast+broadcast. The downside is that non-tree-restricted shortcuts sometimes offer better quality guarantees and therefore better performance.

\subsection{Main result and applications}\label{sec:main-results-and-applications}

The main contribution of the paper is to introduce a general framework for finding \rev{near-optimal tree-restricted} shortcuts in graphs where the only assurance is that they exist. We restate the result.

\begin{theorem*}[Detailed version of \Cref{theorem:shortcuts-given-guarantee}]
Let $\mathcal{P} = (P_1, \ldots, P_N)$ be parts in the graph $G$ with a spanning tree $T \subseteq G$ such that there exists a $T$-restricted shortcut with congestion $c$ and block parameter $b$. There exists a distributed CONGEST algorithm that finds a $T$-restricted shortcut with congestion $O(c \log N)$ and block parameter $3b$ with high probability (with probability at least $1 - n^{-O(1)}$, where any constant can be chosen in the exponent). The shortcut can be found in $O(D \log n\log N + bD \log N + b c\log N)$ rounds.
\end{theorem*}

We note that the Theorems \ref{theorem:shortcuts-given-guarantee} and \ref{theorem:shortcut-existence-on-bounded-genus} immediately give a novel result: an algorithm for constructing shortcuts on bounded genus graphs.
\begin{corollary}
  Given a genus-$g$ graph with diameter $D$ and $N$ parts there is a \rev{(randomized)} distributed algorithm that computes a tree-restricted shortcut with congestion $O(gD\log D\log N)$ and block parameter $O(\log D)$ in $O(gD\log^2D\log N)$ rounds with high probability.
\end{corollary}

Next, we explain how to use tree-restricted shortcuts to distributedly compute the MST on genus-$g$ graphs. Similarly to \cite{gh2016lowcongestion}, we incorporate the shortcuts into the classic 1926 algorithm of Boruvka~\cite{Boruvka26}.

\begin{corollary}
  Given a genus-$g$ graph with $n$ nodes and diameter $D$, there is a \rev{(randomized)} distributed algorithm that computes the Minimum Spanning Tree in $O(gD\log^2D\log^2 n)$ rounds with \rev{high probability}.
\end{corollary}
For completeness we give a brief proof outline:
\begin{proof}
  Boruvka's algorithm runs in $O(\log n)$ phases. Each phase starts with a partition of the graph into connected parts; each part has previously computed the MST on the subgraph induced by the part. Initially, the algorithm starts with the trivial partition in which each node is in its own part. During each phase, each part $P_i$ suggests a merge along the minimum-weighted edge going out of $P_i$. It is well-known that all such edges belong to some MST. By computing a tree-restricted shortcut for each part in $O(gD\log^2D\log n)$ rounds and using our convergecast algorithm on it in $O(gD\log^2 D)$ rounds we can compute the min-weight outgoing edge from each part. A slight difficulty remains: many parts could chain together to form a new part, making the assignment of part IDs in the newly merged part difficulty. This can be avoided by restricting the merge shapes to be star graphs: each part can independently mark itself as a \textbf{head} or \textbf{tail} with probability $\frac{1}{2}$; we are only allowed to merge tails to heads. The number of phases remains $O(\log n)$ as every minimum-weighted outgoing edge will be used for merging with probability at least $\frac{1}{4}$, thus reducing the expected number of parts by a constant.
\end{proof}

\section{Constructing Tree Restricted Shortcuts} \label{sec:algorithm}

In this section, we describe an algorithmic framework that solves the problem of finding near-optimal tree-restricted shortcuts.

\subsection{Overview of the algorithmic framework}

Our algorithm \FnName{FindShortcut} uses two separate subroutines:
\begin{itemize}
\item \textbf{Core:} This subroutine finds a good-quality shortcut with respect to at least a constant fraction of the parts. As a prerequisite, we assume we constructed a tree $T$ with depth $O(D)$ such there exists a $T$-restricted shortcut with congestion $c$ and block parameter $b$. Note that we only assume the \rev{tree-restricted} shortcut's existence.

\begin{lemma}
  \label{lemma:corefast}
  \global\def\StateLemmaCore{
    Let $T$ be a spanning tree with depth $O(D)$ and assume there exists a $T$-restricted shortcut with congestion $c$ and block parameter $b$. The subroutine \FnName{CoreFast} finds a $T$-restricted shortcut $\mathcal{H'} = ( H'_i )_{i=1}^N$ with the following properties:
  \begin{enumerate}
  \item The congestion of $\mathcal{H'}$ is at most $8c$ with high probability.
  \item There exists a subset of parts $\mathcal{P}' \subseteq \mathcal{P}$ with size at least $|\mathcal{P'}| \ge \frac{|\mc{P}|}{2}$ such that each part in $\mc{P}'$ has at most $3b$ block components.
  \end{enumerate}
  The subroutine takes $O(D\log n + c)$ CONGEST rounds to execute \rev{with high probability}. Upon completion, each node knows for each of its incident edges which parts are they assigned to in $\mathcal{H'}$.
}\StateLemmaCore 
\end{lemma}
\rev{We present two versions of the core subroutine for purposes of exposition}. We present a deterministic and simper \FnName{CoreSlow} requiring $O(D\cdot c)$ rounds and a randomized \FnName{CoreFast} requiring $O(D\log n + c)$ rounds. We note that the \FnName{CoreFast} subroutine is the only randomized building block of our framework. Therefore, we can replace it with a deterministic (albeit slower) version at a cost of an additional $\frac{c}{\log n}$ factor.

\item \textbf{Verification:} This subroutine is used to identify the parts $i$ for which the shortcut edges $H_i$ have a sufficiently small number of block components. \rev{The following result follows directly from \Cref{lemma:det-check-if-block-at-most-b}.}
\begin{corollary}
  \label{lemma:verification}
  \global\def\StateLemmaVerification{
  Given a tree $T$ with depth at most $D$ and a tentative $T$-restricted shortcut $\mathcal{H'}$ with congestion $c$, the deterministic subroutine \FnName{Verification} finds all parts $\mathcal{P'} \subseteq \mathcal{P}$ whose designated shortcuts have at most $b'$ block components. The subroutine takes $O(b'(D + c))$ CONGEST rounds to execute. Upon completion, each node knows whether its part is in the set $\mathcal{P'}$ or not.
}\StateLemmaVerification
\end{corollary}
\end{itemize}



We use the subroutines in \FnName{FindShortcut} that implements the construction of \Cref{theorem:shortcuts-given-guarantee}.

\paragraph{Algorithm FindShortcut:} We run the \FnName{CoreFast} subroutine that computes a shortcut $\mathcal{H'} = \{ H'_1, \ldots, H'_N \}$ with congestion $8c$, but possibly an unacceptably large block parameter. The next step is to run the \FnName{Verification} subroutine that finds all parts whose computed shortcut edges $H'_i$ have at most $3b$ block components. We call those parts \textbf{good} and fix their computed shortcut edges and discard the rest. The subroutine is iteratively repeated for $O(\log N)$ rounds at which point the parts have been marked as good.

\global\def\ProofFindShortcuts{
\begin{proof}[Proof of \Cref{theorem:shortcuts-given-guarantee}]
  By \Cref{lemma:corefast}, in each iteration we find a shortcut with congestion $8c$ and block parameter $3b$ for at least half of the parts that have not yet been marked as good, \whp This implies that after $O(\log N)$ iterations all the parts are marked as good. This further implies that the congestion of $\mathcal{H'}$ is $O(c \log N)$ as the congestion of the union of partial shortcuts is at most the sum of congestion of individual partial shortcuts.

  Finally, the number of rounds is at most $O(\log N)$ times the combined number of rounds of the \FnName{CoreFast} and \FnName{Verification} subroutines, namely $O(\log N \cdot (D\log n + c + bD + bc)) = O(D \log N \log n + b D \log N + b c \log N)$ w.h.p.
\end{proof}
}\fullOnly{\ProofFindShortcuts}
\shortOnly{\begin{proof}Deferred to \Cref{sec:construction-proofs}.\end{proof}}


\subsection{Warm-up: an $O(D\cdot c)$-round version of the core subroutine}

In this section, we explain a simple and deterministic, but slower version of the core subroutine named \FnName{CoreSlow} that terminates in $O(D \cdot c)$ CONGEST rounds. We improve its round complexity to $O(D\log n + c)$ in the following section.

On a high level, the subroutine takes each part $P_i$ and tries to assign the $T$-ancestors of nodes in $P_i$ to its shortcut edges $H'_i$. However, this might lead to a large congestion on some edges. We address this issue by declaring an edge \textbf{unusable} if more than $2c$ different parts try to use it. This ensures the congestion is at most $2c$. We show the process provably leads to a constant fraction parts having small congestion and a small block parameter.

\textbf{Preliminaries:} As standard, assume we fix a spanning tree $T = (V, E_T)$ of depth $O(D)$ such that $G$ has a $T$-restricted shortcut with congestion $c$ and block parameter $b$. During the execution of the algorithm, some of the edges will be marked as \textbf{unusable}. Furthermore, we say that a tree edge $e \in E_T$ \textbf{can see} a node $v \in V$ if $v$ is in the subtree of $e$ and no edge on the unique path between the lower endpoint of $e$ and $v$ is unusable. Analogously, an edge can see a part $P_i$ if it can see any node in $P_i$.

\textbf{Outline of the \FnName{CoreSlow} subroutine:} Initially, no edge is unusable. We process the (tree) edges of $T$ in order of decreasing depth (bottom to top). An edge $e$ is assigned to all parts $P_i$ that $e$ can see. If an edge is assigned to more than $2c$ different parts, we mark this edge $e$ as \textbf{unusable} disallow $e$ from being used at all by any part.

\textbf{A detailed description of the \FnName{CoreSlow} subroutine:} Each node $v$ maintains a list $L_v$ of part IDs that $v$'s $T$-parent edge can see. The lists $L_v$ are initially empty. The subroutine runs in $\mathrm{depth}(T)$ phases where in phase $k$ each node $v$ at depth $\mathrm{depth}(T) - k$ updates $L_v$ simultaneously and sends the entire list $L_v$ to its ($v$'s) $T$-parent. Consider a node $v$ that receives $L_{v'}$ for all its $T$-children $v'$. We assign the union of all received lists and the singleton part ID of $v$ (if any) to $L_v$. If $|L_v| \le 2c$, we assign the parent edge of $v$ to all the parts in $L_v$ and transmit $L_v$ to its parent (potentially requiring $2c$ rounds). Otherwise, if $|L_v| > 2c$, we declare the parent edge as unusable.

A direct implementation of this would lead to a subroutine that takes $O(D \cdot c)$ rounds in the CONGEST model. Each of the $O(D)$ levels of $T$ must propagate at most $2c$ part IDs to their parent nodes. However, this bottleneck can be improved by random sampling, as we show in the next section with the subroutine \FnName{CoreFast}. \shortOnly{High-level pseudocode is given in \Cref{sec:pseudocodes}.}

\global\def\AlgCoreSlow{\begin{algorithm}[H]
  \caption{\FnName{CoreSlow}}
  \label{alg:coreslow}
  \begin{enumerate}
  \item At time $k$ each node $v$ at depth $depth(T) - k$ does the following in parallel:
    \begin{enumerate}
    \item if $v$ is an element of $P_i$, set $L_v \gets \{ i \}$, otherwise $L_v \gets \emptyset$
    \item receive all the part IDs from $v$'s children and assign their union to $L'$
    \item $L_v \gets L_v \cup L'$
    \item if $|L_v| > 2c$, mark $v$'s parent edge as unusable
    \item otherwise, (serially) send all the part IDs of $L_v$ up to $v$'s parent node
    \end{enumerate}
  \item For each node $v$:
    \begin{enumerate}
    \item if the parent edge $e$ of $v$ is marked as unusable, $e$ will not be assigned to any part
    \item otherwise, $e$ will be assigned to all $H_i, \forall i \in L_v$
    \end{enumerate}
  \end{enumerate}
\end{algorithm}}
\fullOnly{\AlgCoreSlow}

\begin{lemma}
  \label{lemma:coreslow}
\global\def\StateLemmaCoreSlow{
  Let $T$ be a spanning tree of depth $O(D)$ and assume there exists a $T$-restricted shortcut with congestion $c$ and block parameter $b$. The subroutine \FnName{CoreSlow} finds a $T$-restricted shortcut $\mathcal{H'} = ( H'_1, H'_2, ..., H'_N )$ with the following properties:
  \begin{enumerate}
  \item \label{lemma:coreslow:propcongest} The congestion of $\mathcal{H'}$ is at most $2c$.
  \item \label{lemma:coreslow:propblock} There exists a subset of parts $\mathcal{P}' \subseteq \mathcal{P}$ with size at least $|\mathcal{P'}| \ge \frac{|\mc{P}|}{2}$ such that each part in $\mc{P}'$ has at most $3b$ block components.
  \end{enumerate}
  The subroutine is deterministic and takes $O(D\cdot c)$ CONGEST rounds to execute. Upon completion, each node knows for each of its incident edges which parts are they assigned to in $\mathcal{H'}$.
}\StateLemmaCoreSlow
\end{lemma}
\global\def\ProofCoreSlow{\begin{proof}
  Let $\mathcal{H} = ( H_i )$ be any $T$-restricted shortcut with congestion $c$ and block parameter $b$ and let $\mathcal{H'} = ( H'_i )$ be the shortcut computed by \FnName{CoreSlow}. We call $\mathcal{H}$ the \textbf{canonical} shortcut and $\mathcal{H'}$ the \textbf{computed} shortcut.

  By construction, the congestion of $\mathcal{H'}$ is $2c$ as any edge that would be assigned to more than $2c$ parts is marked as unusable. Hence we proved property \ref{lemma:coreslow:propcongest}.

  Let $U \subseteq E_T$ be the set of unusable edges marked by the subroutine. In this paragraph, we find an upper bound for $|U|$. Consider \textbf{blaming} a part $P_i$ for congesting an unusable edge $e \in U$ when $e \not \in E_G[P_i] \cup H_i$ and $e$ can see $P_i$, i.e.,  edge $e$ was not in the canonical shortcut $H_i$, but $e$ was congested by part $P_i$ (and ultimately declared unusable). Each part can be blamed at most $b$ times because each block component can only be blamed for the first unusable edge in his $T$-tree path towards the $T$-root. Furthermore, if $e$ is unusable, it takes at least $2c - c$ different block components (from different parts) to be blamed for congesting $e$. Therefore $|U| \le N \frac{b}{c}$.

  We say that a part $P_i$ \textbf{missed} an edge $e$ when $e \in E_G[P_i] \cup H_i$ and $e \in U$ (consequently, $e \not \in H'_i$). Furthermore, call a part \textbf{bad} if it missed at least $2b$ edges and \textbf{good} otherwise. Note that if a part $P_i$ is good, the block parameter of $H'_i$ is at most $2b + \text{blockParameter}(\mathcal{H}) = 3b$. This is because each missed edge induces a new block component in $\mathcal{H'}$ (more precisely, we can identify each block component of $\mathcal{H'}$ with either a unique block component of $\mathcal{H}$ or a unique missed edge $e \in U$). Consequently, it is sufficient to prove that the subroutine finds at least $\frac{1}{2}N$ good parts.

  As any unusable edge is assigned to at most $c$ parts in the canonical shortcut, and for a part to be bad we need at least $2b$ edges to be missed, we have that the number of bad parts is at most $|U| \frac{c}{2b} \le \frac{1}{2}N$. Hence, the subroutine finds at least $\frac{1}{2}N$ good shortcuts, proving property \ref{lemma:coreslow:propblock}.

  The subroutine terminates in $O(D \cdot c)$ rounds: on each of the $O(D)$ levels of the tree $T$, all the nodes in parallel must send the part IDs trying to use its parent edge up the tree. A node can send up to $2c$ IDs, each requiring one round for its transmission.
\end{proof}}
\fullOnly{\ProofCoreSlow}
\shortOnly{\begin{proof}Deferred to \Cref{sec:construction-proofs}.\end{proof}}

\subsection{A faster $O(D\log n + c)$-round version of the core subroutine}

In this section, we describe a faster version of the core subroutine named \FnName{CoreFast}. On a high level, we lower the running time of \FnName{CoreSlow} by estimating the number of parts trying to use an edge by random sampling. In particular, each part becomes \textbf{active} with probability $p$ and we declare an edge unusable when $\Omega(c \cdot p)$ active parts try to use that edge.

\textbf{Preliminaries:} In addition to the preliminaries of \FnName{CoreSlow} we need shared randomness between all the nodes within a part. In other words, all the nodes of the same part must have access to the same seeds for a pseudorandom generator. This can be done by sharing $O(\log^2 n)$ random bits among all the nodes of $G$ in $O(D + \log n)$ rounds, as described in \cite{gh2016lowcongestion}.

\textbf{Outline of the \FnName{CoreFast} subroutine:} Each part becomes \textbf{active} with probability $p = \frac{\gamma \log n}{2c}$ where $\gamma > 0$ is sufficiently large constant. We basically follow the \FnName{CoreSlow} subroutine, but instead of propagating all $O(c)$ part IDs of $L_v$, we propagate only the active ones. An edge is declared \textbf{unusable} if at least $4c \cdot p = \Omega(\log n)$ (active) part IDs want to use it. Hence, by a standard Chernoff bound argument we can claim with high probability that (i) we never propagate more than $O(\log n)$ part IDs through an edge, (ii) each unusable edge has at least $2c$ part IDs trying to use that edge, and (iii) each usable (non-congested) edge has at most $8c$ part IDs. After determining which edges are unusable in $O(D\log n)$ rounds, \FnName{CoreFast} must nevertheless find the complete set of part IDs that can use each edge. This is a tree routing problem where each message (part ID) has to be routed up the tree $T$ until the first unusable edge. No message needs to travel more than $D$ edges and no edge needs to transmit more than $8c$ different part IDs \whp Hence this routing can be done in $O(D + c)$ using \Cref{lemma:routing-on-trees}.

\textbf{A detailed description of the \FnName{CoreFast} subroutine:} Due to shared randomness, each part independently becomes \textbf{active} with probability $p = \frac{\gamma \log n}{2c}$ (all the nodes within the part agree on this label). Similarly, as in \FnName{CoreSlow}, each node $v$ maintains a list $\tilde{L}_v$ of active part IDs that its ($T$) parent edge can see. The lists $\tilde{L}_v$ are initially empty. The subroutine runs in $\mathrm{depth}(T)$ phases where in phase $k$ all the nodes at depth $\mathrm{depth}(T) - k$ try to update $\tilde{L}_v$ in parallel and send $\tilde{L}_v$ to its $T$-parent. Consider a node $v$ that receives $L_{v'}$ for all its $T$-children $v'$. We assign the union of all received lists and the singleton part ID of $v$ (if any) to $L_v$. If $|L_v| \le 4c \cdot p$, we assign the parent edge of $v$ to all the parts in $L_v$ and transmit $L_v$ to its parent (potentially requiring $O(\log n)$ rounds). Otherwise, if $|L_v| > 4 \cdot p$, we declare the parent edge as unusable. This finalizes the first part of the subroutine where we determine all unusable edges. It remains to forward the complete set of part IDs (and not just the sampled ones) that can use some edge $e$ to the endpoints of $e$. This is a classic tree routing problem where no route has its length larger than $D$ and no edge intersects more than $8c$ paths \whp \Cref{lemma:routing-on-trees} provides a method to route all part IDs in at most $O(D + c)$ rounds. Note that any two part IDs whose routes share an edge have the same endpoint (lowest unusable ancestor edge), so any routing priority between the messages gives the aforementioned $O(D + c)$ bound \whp \shortOnly{High-level pseudocode is given in \Cref{sec:pseudocodes}.}

\global\def\AlgCoreFast{\begin{algorithm}[H]
  \caption{\FnName{CoreFast}}
  \label{alg:corefast}
  \begin{enumerate}
  \item Each part becomes active with probability $p = \frac{\gamma \log n}{2c}$
  \item At time $k$ each node $v$ at depth $depth(T) - k$ does the following in parallel:
    \begin{enumerate}
    \item if $v$ is an element of $P_i$ and $P_i$ is active, set $\tilde{L}_v \gets \{ i \}$, otherwise $\tilde{L}_v \gets \emptyset$

    \item receive all the active part IDs from $v$'s children and assign their union to $L'$
    \item $\tilde{L}_v \gets \tilde{L}_v \cup L'$
    \item if $|\tilde{L}_v| \ge 4c \cdot p$, mark $v$'s parent edge as unusable

    \item otherwise send all the part IDs $\tilde{L}_v$ up to $v$'s parent node
    \end{enumerate}

  \item Each node $v$ initializes $Q_v$ with its part ID (or $\emptyset$ if not in any part)
  \item Each node $v$ does the following in parallel:
    \begin{enumerate}
      \item add all received IDs to the $Q_v$
      \item if the parent edge of $v$ is not unusable and $\exists i \in Q_v$ that was never forwarded
        \begin{enumerate}
          \item forward minimum such $i$ along the parent edge
        \end{enumerate}
      \end{enumerate}
    \item Each part ID in $Q_v$ can use the parent edge of $v$ unless it is unusable
    \end{enumerate}
  \end{algorithm}}
\fullOnly{\AlgCoreFast}

\begin{lemma*}[Restated \Cref{lemma:corefast}]
  \StateLemmaCore
\end{lemma*}
\global\def\ProofLemmaCoreFast{\begin{proof}
  This proof extensively utilizes methods used in the proof of \Cref{lemma:coreslow}. For completeness, we redefine all of the used terminologies and reprove all of the intermediate results.

  Let $\mathcal{H} = ( H_i )$ be any $T$-restricted shortcut with congestion $c$ and block parameter $b$ and let $\mathcal{H'} = ( H'_i )$ be the shortcut computed by \FnName{CoreFast}. We call $\mathcal{H}$ the \textbf{canonical} shortcut and $\mathcal{H'}$ the \textbf{computed} shortcut.

  \rev{Consider any tree edge. Suppose that the edge can see $t$ different part IDs. Denote by $X_1, \ldots, X_t$ whether those $t$ parts are active (in which case $X_i = 1$, otherwise $X_i = 0$). Let $S := X_1 + X_2 + \ldots + X_t$. Due to sampling, we have that the expectation $\mathbb{E}[S] = p t$. Since $X_i \in \{0, 1\}$ and they are independent we can apply a standard Chernoff bound argument giving us that $\Pr[X_1 + \ldots + X_t \le \frac{1}{2} \mathbb{E}[S] ] \le \exp(-\delta \mathbb{E}[S])$ for some constant $\delta > 0$. Suppose now that $t \ge 8c$, we have that $\Pr[X_1 + \ldots + X_t \le 4c \cdot p] \le \exp(-\delta 8 p c) = \exp(-\delta 4 \gamma \log) = n^{-\gamma'}$ for a sufficiently large constant $\gamma' > 0$ (since we choose $\gamma > 0$ sufficiently large). We conclude that if $t \ge 8c$, the considered edge will become unusable with high probability. Since there are only a polynomial number of different edges, we can use a union bound to conclude that the congestion of $\mathcal{H'}$ is $8c$ (for all edges) with high probability (since the probability of this being violated is at most $n \cdot n^{-\gamma'} = n^{-\gamma + 1}$, i.e., with high probability).}


  Let $U \subseteq E_T$ be the set of unusable edges marked by the subroutine. In this paragraph, we find an upper bound for $|U|$. Consider \textbf{blaming} a part $P_i$ for congesting an unusable edge $e \in U$ when $e \not \in E_G[P_i] \cup H_i$ and $e$ can see $P_i$, i.e., edge $e$ was not in the canonical shortcut $H_i$, but $e$ was congested by part $P_i$ (and ultimately declared unusable). We argue via a Chernoff bound that each unusable edge $e \in U$ can see at least $2c$ parts.

  \rev{The bound is argued in a completely analogous way as proving the congestion being at most $8c$, except the Chernoff bound we use here is the following one. As before, let $S := X_1 + X_2 + \ldots X_t$ be the sum of indicator variables of the part IDs that can see an edge the fixed edge $e$. Our bound stipulates that for independent $\{0,1\}$ variables $X_i$ we have that $\Pr[X_1 + X_2 + \ldots + X_t \le 2 \mathbb{E}[S]] \le \exp(- \delta \mathbb{E}[S])$ for some $\delta > 0$. Using it, we conclude that if $t \le 2c$ parts can see $e \in U$, then $\Pr[S \ge 2\mathbb{E}[S]] = \Pr[S \ge 4c \cdot p] \le n^{-\gamma'}$ for some sufficiently large $\gamma' > 0$, giving us that in such a case the would not be declared unusable with high probability. Union bounding, we get the same holds for each $e \in U$.}
  
  Since each unusable edge $e \in U$ can see at least $2c$ parts, we blame at least $2c - \mathrm{congestion}(\mathcal{H}) = c$ parts for congesting $e$. Each part can be blamed at most $b$ times because each block component can only be blamed for the first unusable edge in his $T$-tree path towards the $T$-root. Furthermore, if $e$ is unusable, it takes at least $2c - c$ different block components (from different parts) to be blamed for congesting $e$. Therefore $|U| \le N \frac{b}{c}$.

  We say that a part $P_i$ \textbf{missed} an edge $e$ when $e \in E_G[P_i] \cup H_i$ and $e \in U$ (consequently $e \not \in H'_i$). Furthermore, call a part \textbf{bad} if it missed at least $2b$ edges and \textbf{good} otherwise. Note that if a part $P_i$ is good, the block parameter of $H'_i$ is at most $2b + \text{blockParameter}(\mathcal{H}) = 3b$. This is because each missed edge induces a new block component in $\mathcal{H'}$ (more precisely, we can identify each block component of $\mathcal{H'}$ by either a unique block component of $\mathcal{H}$ or a unique missed edge $e \in U$). Consequently, it is sufficient to prove that the subroutine finds at least $\frac{1}{2}N$ good parts.

  As any unusable edge is assigned to at most $c$ parts in the canonical shortcut and for a part to be bad we need at least $2b$ edges to be missed, we have that the number of bad parts is at most $|U| \frac{c}{2b} \le \frac{1}{2}N$. Hence, the subroutine finds at least $\frac{1}{2}N$ good shortcuts.

  The subroutine takes $O(D\log n + c)$ rounds: on each of the $O(D)$ levels of the tree $T$, all the nodes in parallel must send the active part IDs that its parent edge can see. If an edge $e$ is not unusable, we argued via a Chernoff bound that at most $O(c \cdot p) = O(\log n)$ active part IDs can be seen from $e$, hence the number of rounds for determining unusable edges is $O(D\log n)$, w.h.p. Finally, propagating the part IDs upwards along $T$ described in \Cref{lemma:routing-on-trees} takes $O(D + c)$ rounds, bringing the total number of rounds to $O(D\log n + c)$.
\end{proof}}
\fullOnly{\ProofLemmaCoreFast}
\shortOnly{\begin{proof}Deferred to \Cref{sec:construction-proofs}.\end{proof}}

\bibliographystyle{alpha}
\bibliography{shortcuts}

\newcommand{\etalchar}[1]{$^{#1}$}
\begin{thebibliography}{DSHK{\etalchar{+}}11}

\bibitem[DSHK{\etalchar{+}}11]{DasSarma-11}
Atish Das~Sarma, Stephan Holzer, Liah Kor, Amos Korman, Danupon Nanongkai,
  Gopal Pandurangan, David Peleg, and Roger Wattenhofer.
\newblock Distributed verification and hardness of distributed approximation.
\newblock In {\em Proc.\ of the Symp.\ on Theory of Comp.\ (STOC)}, pages
  363--372, 2011.

\bibitem[Elk04]{Elkin-2004}
Michael Elkin.
\newblock Unconditional lower bounds on the time-approximation tradeoffs for
  the distributed minimum spanning tree problem.
\newblock In {\em Proc.\ of the Symp.\ on Theory of Comp.\ (STOC)}, pages
  331--340, 2004.

\bibitem[Elk06]{elkin2006unconditional}
Michael Elkin.
\newblock An unconditional lower bound on the time-approximation trade-off for
  the distributed minimum spanning tree problem.
\newblock {\em SIAM Journal on Computing}, 36(2):433--456, 2006.

\bibitem[FHW12]{Frischknecht-Diameter-2012}
Silvio Frischknecht, Stephan Holzer, and Roger Wattenhofer.
\newblock Networks cannot compute their diameter in sublinear time.
\newblock In {\em Proc.\ of ACM-SIAM Symp.\ on Disc.\ Alg.\ (SODA)}, pages
  1150--1162, 2012.

\bibitem[GH15]{GH15-Embedding}
Mohsen Ghaffari and Bernhard Haeupler.
\newblock Distributed algorithms for planar networks {I}: Planar embedding.
\newblock Manuscript, 2015.

\bibitem[GH16]{gh2016lowcongestion}
Mohsen Ghaffari and Bernhard Haeupler.
\newblock Distributed algorithms for planar networks {{I}{I}}: Low-congestion
  shortcuts, mst, and min-cut.
\newblock In {\em Proc.\ of ACM-SIAM Symp.\ on Disc.\ Alg.\ (SODA)}, pages
  202--219. SIAM, 2016.

\bibitem[GH20]{ghaffari2020excluding}
Mohsen Ghaffari and Bernhard Haeupler.
\newblock Low-congestion shortcuts for graphs excluding dense minors.
\newblock 2020.

\bibitem[GK13]{Ghaffari-Kuhn}
Mohsen Ghaffari and Fabian Kuhn.
\newblock Distributed minimum cut approximation.
\newblock In {\em Proc.\ of the Int'l Symp.\ on Dist.\ Comp.\ (DISC)}, pages
  1--15, 2013.

\bibitem[GKK{\etalchar{+}}15]{Ghaffari:2015}
Mohsen Ghaffari, Andreas Karrenbauer, Fabian Kuhn, Christoph Lenzen, and Boaz
  Patt-Shamir.
\newblock Near-optimal distributed maximum flow: Extended abstract.
\newblock In {\em the Proc.\ of the Int'l Symp.\ on Princ.\ of Dist.\ Comp.\
  (PODC)}, pages 81--90, 2015.

\bibitem[GKP93]{Garay-Kutten-Peleg}
J.A. Garay, S.~Kutten, and D.~Peleg.
\newblock A sub-linear time distributed algorithm for minimum-weight spanning
  trees.
\newblock In {\em Proc.\ of the Symp.\ on Found.\ of Comp.\ Sci.\ (FOCS)},
  1993.

\bibitem[GL18]{ghaffari2018new}
Mohsen Ghaffari and Jason Li.
\newblock New distributed algorithms in almost mixing time via transformations
  from parallel algorithms.
\newblock {\em arXiv preprint arXiv:1805.04764}, 2018.

\bibitem[HHW18]{haeupler2018round}
Bernhard Haeupler, D~Ellis Hershkowitz, and David Wajc.
\newblock Round-and message-optimal distributed graph algorithms.
\newblock In {\em Proceedings of the 2018 ACM Symposium on Principles of
  Distributed Computing}, pages 119--128. ACM, 2018.

\bibitem[HIZ16a]{haeupler2016low}
Bernhard Haeupler, Taisuke Izumi, and Goran Zuzic.
\newblock Low-congestion shortcuts without embedding.
\newblock In {\em Proceedings of the 2016 ACM Symposium on Principles of
  Distributed Computing}, pages 451--460. ACM, 2016.

\bibitem[HIZ16b]{haeupler2016near}
Bernhard Haeupler, Taisuke Izumi, and Goran Zuzic.
\newblock Near-optimal low-congestion shortcuts on bounded parameter graphs.
\newblock In {\em International Symposium on Distributed Computing}, pages
  158--172. Springer, 2016.

\bibitem[HL18]{haeupler2018faster}
Bernhard Haeupler and Jason Li.
\newblock Faster distributed shortest path approximations via shortcuts.
\newblock {\em arXiv preprint arXiv:1802.03671}, 2018.

\bibitem[HLZ18]{haeupler2018minor}
Bernhard Haeupler, Jason Li, and Goran Zuzic.
\newblock Minor excluded network families admit fast distributed algorithms.
\newblock In {\em Proceedings of the 2018 ACM Symposium on Principles of
  Distributed Computing}, pages 465--474. ACM, 2018.

\bibitem[HW12]{Holzer-Paths-2012}
Stephan Holzer and Roger Wattenhofer.
\newblock Optimal distributed all pairs shortest paths and applications.
\newblock In {\em the Proc.\ of the Int'l Symp.\ on Princ.\ of Dist.\ Comp.\
  (PODC)}, pages 355--364, 2012.

\bibitem[IW14]{Izumi2014}
Taisuke Izumi and Roger Wattenhofer.
\newblock Time lower bounds for distributed distance oracles.
\newblock In {\em Proc. of the International Conference on Principles of
  Distributed Systems}, pages 60--75, 2014.

\bibitem[KP95]{Kutten-Peleg}
Shay Kutten and David Peleg.
\newblock Fast distributed construction of k-dominating sets and applications.
\newblock In {\em the Proc.\ of the Int'l Symp.\ on Princ.\ of Dist.\ Comp.\
  (PODC)}, pages 238--251, 1995.

\bibitem[KP08]{Khan2008}
Maleq Khan and Gopal Pandurangan.
\newblock A fast distributed approximation algorithm for minimum spanning
  trees.
\newblock {\em Distributed Computing}, 20(6):391--402, 2008.

\bibitem[LMR94]{LMR94-routing}
Frank~Thomson Leighton, Bruce~M Maggs, and Satish~B Rao.
\newblock Packet routing and job-shop scheduling in {O}(congestion+ dilation)
  steps.
\newblock {\em Combinatorica}, 14(2):167--186, 1994.

\bibitem[LPSP19]{lenzen2019distributed}
Christoph Lenzen, Boaz Patt-Shamir, and David Peleg.
\newblock Distributed distance computation and routing with small messages.
\newblock {\em Distributed Computing}, 32(2):133--157, 2019.

\bibitem[Nan14]{Nanongkai-paths}
Danupon Nanongkai.
\newblock Distributed approximation algorithms for weighted shortest paths.
\newblock In {\em Proc.\ of the Symp.\ on Theory of Comp.\ (STOC)}, pages
  565--573, 2014.

\bibitem[NMN01]{Boruvka26}
Jaroslav Ne{\v{s}}et{\v{r}}il, Eva Milkov{\'a}, and Helena
  Ne{\v{s}}et{\v{r}}ilov{\'a}.
\newblock Otakar boruvka on minimum spanning tree problem translation of both
  the 1926 papers, comments, history.
\newblock {\em Discrete Math.}, 233(1):3--36, 2001.

\bibitem[NS14]{nanongkai2014almost}
Danupon Nanongkai and Hsin-Hao Su.
\newblock Almost-tight distributed minimum cut algorithms.
\newblock In {\em Proc.\ of the Int'l Symp.\ on Dist.\ Comp.\ (DISC)}, pages
  439--453, 2014.

\bibitem[Pel00]{Peleg:2000}
David Peleg.
\newblock {\em Distributed Computing: A Locality-sensitive Approach}.
\newblock Society for Industrial and Applied Mathematics, Philadelphia, PA,
  USA, 2000.

\bibitem[PR99]{Peleg-Rubinovich-1999}
David Peleg and Vitaly Rubinovich.
\newblock A near-tight lower bound on the time complexity of distributed
  {{M}{S}{T}} construction.
\newblock In {\em Proc.\ of the Symp.\ on Found.\ of Comp.\ Sci.\ (FOCS)},
  pages 253--, 1999.

\end{thebibliography}

\appendix

\end{document}